\documentclass[final,3p,times,twocolumn]{elsarticle}

\usepackage[utf8]{inputenc}
\usepackage{xspace}
\newcommand{\spin}{SPIN\xspace}
\newcommand{\rev}[1]{#1}   % revision markup disabled for submission
\usepackage{amsmath,amssymb,amsthm}
\usepackage{url}
\usepackage{graphicx}
\usepackage{booktabs}
\usepackage{enumitem}
\usepackage{xcolor}
\newcommand{\fix}[1]{#1}  % round-2 corrections; set to \textcolor{red}{#1} to re-enable highlighting
\newcommand{\new}[1]{\textcolor{black}{#1}}  % round-3 markup
\setlist[itemize]{itemsep=1pt, topsep=2pt, parsep=0pt, partopsep=0pt}

\newtheorem{proposition}{Proposition}
\newtheorem{lemma}{Lemma}
\newtheorem{assumption}{Assumption}
\newtheorem{remark}{Remark}
\newtheorem{definition}{Definition}

\journal{Journal of Process Control}

\begin{document}

\begin{frontmatter}

\title{Model-Free PID Tuning by Step-Response Inspection}

\author[inst1]{Daniel Pachner}
\author[inst1]{Pavel Otta\corref{cor1}}
\ead{pavel.otta@cvut.cz}
\author[inst1]{Ji\v{r}\'{\i} Dost\'al}
\author[inst2]{Vladim\'{\i}r Havlena}

\cortext[cor1]{Corresponding author}

\affiliation[inst1]{organization={University Centre for
Energy Efficient Buildings (UCEEB), Czech Technical University in
Prague}, city={Buštěhrad}, country={Czech Republic}}
\affiliation[inst2]{organization={Department of Control Engineering,
Faculty of Electrical Engineering, Czech Technical University in Prague},
city={Prague}, country={Czech Republic}}

\begin{abstract}
Industrial PID loops are still tuned by hand. The engineer steps the
setpoint, looks at the response, and changes a gain. This paper turns
that procedure into an algorithm, \new{\spin, for Step-response
Phase-portrait INspection}. From a single closed-loop step response we
form three phase portraits of the control deviation, and for each we
count how many times the trajectory winds around its settling point.
The count has no units, and it does not depend on the size of the step
or on the clock. It shows in which frequency band the loop is
under-damped: that of the integral, the proportional, or the derivative
channel. A triangular rule reads the three counts. It reduces the gain
of the one band whose fault is certain, and it raises the gain of every
band that the same test has shown to be within its limit. No model is
identified and no optimization is solved. The iteration is
deterministic, and its gains stay inside a fixed box. Cutting the gain of
a band that rings does reduce its count, because a ringing count is
itself the evidence that the loop is in the regime where this holds. The
iteration therefore settles into a cycle at the edge of the well-damped
set. It needs no excitation beyond a
routine setpoint step, so it can run on a process in operation, and it
may start from arbitrary gains, including destabilizing ones.
\new{\spin} assumes a stable, self-regulating process whose gain and
phase both fall with frequency, \fix{which covers the lag-plus-dead-time
models common in process control,} and a controller from the nested
family
$\mathrm{I} \subset \mathrm{PI} \subset \mathrm{PID}$ with integral
action present. We validate it on a battery of process-typical plants,
using one fixed step size and one set of dimensionless constants
throughout. Every decision is visible: an operator who watches the
portraits sees what the algorithm sees. \new{A reference implementation
and an interactive demonstration are public.}
\end{abstract}

\begin{keyword}
PID control \sep controller tuning \sep data-driven control \sep
phase plane \sep step response
\end{keyword}

\end{frontmatter}

%=====================================================================
\section{Introduction}
%=====================================================================

More than eighty years after Ziegler and Nichols \cite{ZN1942}, PID
remains the dominant industrial controller
\cite{AstromHagglund2001}: in a survey of more than eleven thousand
controllers, $97\%$ of the regulatory loops used PID
\cite{Desborough2002}. Most PID
loops are still tuned by inspection: an engineer steps the setpoint,
looks at the response, and adjusts a gain by judgment.

Three lines of work replace this inspection.
1)~Model-based rules use an identified low-order model instead
\cite{ZN1942,CohenCoon1953,RiveraIMC,Skogestad2003}, and therefore
inherit the errors of that model; both experimental parts of
Ziegler--Nichols have since been revisited with modern robustness tools
\cite{AstromHagglund2004,Schlegel2002,SchlegelCech2005}.
2)~Data-driven schemes replace the inspection with a criterion, and
with the machinery needed to serve that criterion: a dedicated relay
experiment \cite{AstromHagglund1984}, gradient experiments on a
quadratic cost \cite{HjalmarssonIFT}, an identification-like problem
\cite{CampiVRFT}, or online perturbation of the gains
\cite{KillingsworthKrstic}.
\fix{3)~Parametric-region methods compute the admissible set directly
in the plane of two controller parameters, every setting in that set
meeting an $H_\infty$ specification
\cite{BrabecSchlegel2023,Ho2003}. The result is a region with a
certificate rather than a single tuning, but it needs a plant model, a
weighting function and a level fixed in advance, and it treats two
parameters at a time.}

The closest precedent automated the inspection but did not formalize
it: the pattern-recognition autotuner of \cite{KrausMyron1984}, sold
industrially as EXACT, reads overshoot and decay ratio from detected
peaks and adjusts the gains by heuristic logic. The inspection itself
has never been made a well-defined quantity: which gain is at fault,
given this response, computed rather than matched to a stored
pattern.

The inspection differs from all of these in the kind of quantity it
reads. Every diagnostic named above is metric: a ratio of amplitudes,
an integrated cost, or an identified parameter with units. Each of them
therefore needs a scale, and that scale must be supplied for every loop
by a model, by a normalization, or by an expert. The turn index reads the oriented shape of the error trajectory
instead: how many times that trajectory wraps around its settling point.
Being a count of turns, it is topological rather than metric, and so
independent of units, of amplitude and of the clock, and, in the fast-sampling limit, of the sampling period
(Proposition~\ref{prop:invariance}). This is why the manual procedure
is blind to scale. It is also why one fixed set of dimensionless
limits serves every plant in this paper.

This paper formalizes the inspection. \new{We call the resulting
method \spin, for Step-response Phase-portrait INspection. The name
stands for the complete procedure: the three portraits, the guarded
turn indices computed from them, and the triangular rule that acts on
the three counts.} The contributions are the following.
(i)~A triple of phase portraits of the control deviation, whose turn
indices form a damping diagnostic that has no units, needs no scale,
selects one frequency band, and is topological
(Section~\ref{sec:index}).
(ii)~A well-posedness analysis of the index, containing a
counterexample in which a naively windowed count diverges and a
settling-anchored guard that restores independence of the window
(Section~\ref{sec:wellposed}).
(iii)~A single triangular tuning rule in band-ordered gain coordinates,
read as loop shaping band by band. The gains stay bounded without
further conditions, and convergence to a bounded neighbourhood of the
boundary follows from a monotonicity property that the plant class
supplies wherever the rule corrects a gain
(Proposition~\ref{prop:mono}, Section~\ref{sec:rule}).
(iv)~Validation on a battery of process-typical plants, with the same
default parameters throughout (Section~\ref{sec:battery}).

\new{\spin} returns a fast, well-damped tuning, which is what an expert
tuner also offers, and asks for one closed-loop setpoint step per
iteration and nothing else: no model, no relay experiment, no solver. It
can be audited in full, since an operator watching the three portraits
sees what the algorithm sees, and it is governed throughout by the fixed
dimensionless defaults of Table~\ref{tab:defaults}, none of which was
adjusted for any plant.

Every move either repairs a diagnosed fault or, when no band objects,
raises all gains by one bounded step, so no signal is ever injected
only to acquire information, as it is in iterative feedback tuning or in
extremum seeking \cite{HjalmarssonIFT,KillingsworthKrstic}. The
expanding move does push the loop towards its stability boundary, one
step at a time, and the stability screen of
Section~\ref{sec:wellposed} is what keeps that push bounded.

Figures~\ref{fig:init_step} and~\ref{fig:init_portraits} show the
method on one plant (P2, Section~\ref{sec:battery}) before any
mathematics: the response and the portraits before and after tuning,
from a severely mis-tuned start, recorded from the running
implementation (Data availability). The run begins at unit gains
$\new{(K_i, K_p, K_d)} = (1, 1, 1)$ with the derivative on the
measurement, and the P2 preset of the interactive demonstration
reproduces it. The battery of Section~\ref{sec:battery} is a separate
set of runs, starting from a mis-tuned AMIGO design with the derivative
on the error.

\begin{figure*}[t]
\centering
\includegraphics[width=0.47\textwidth]{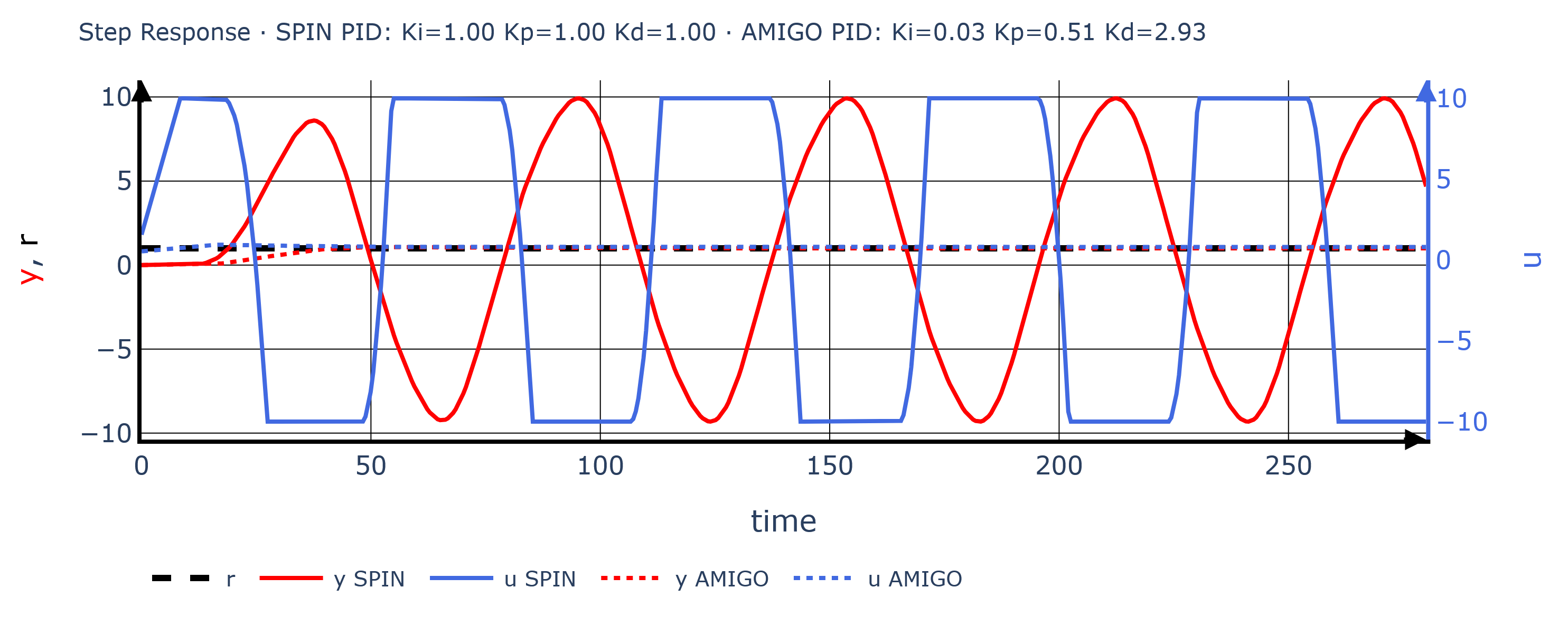}\hfill
\includegraphics[width=0.47\textwidth]{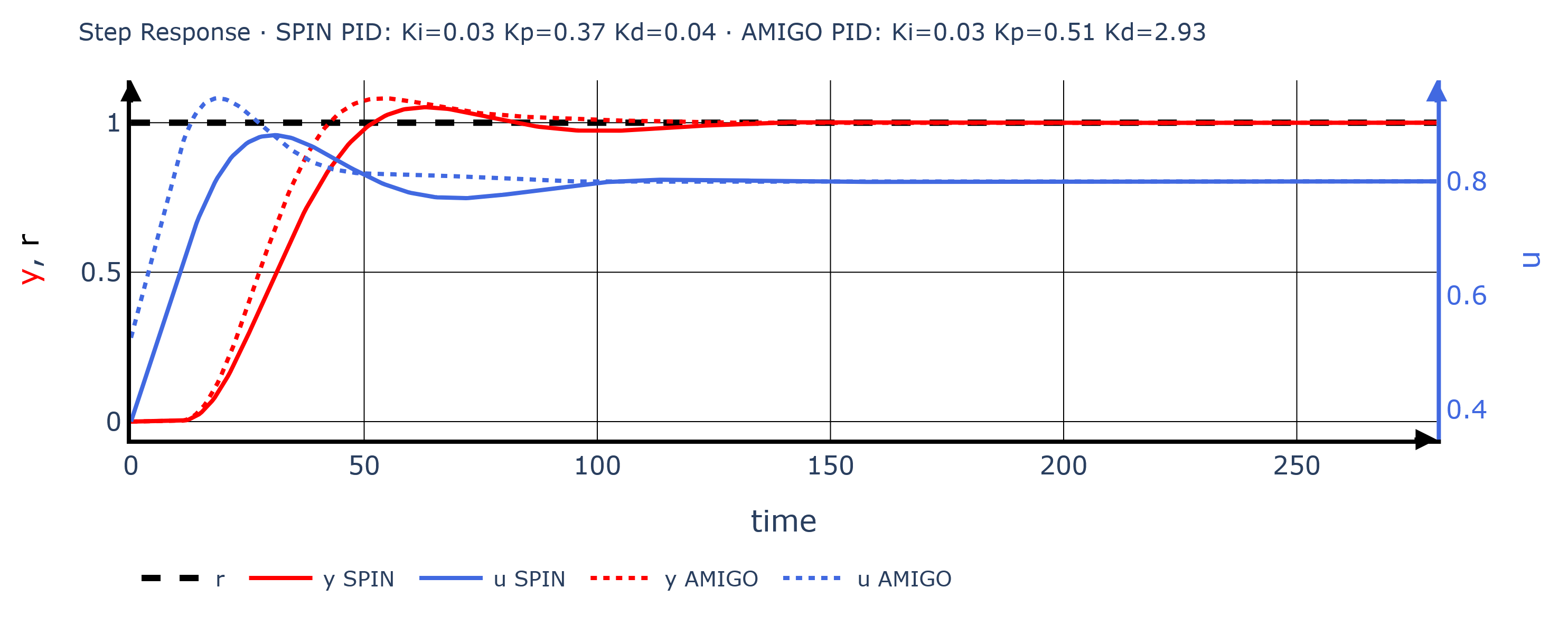}
\caption{Step response before (left) and after (right) tuning; P2 from
unit gains. The mis-tuned loop sustains an oscillation, and the
actuator swings between its limits. Nothing in the left plot shows
which gain is at fault. At the end of the run a single overshoot
settles cleanly, and the actuator stays unsaturated.}
\label{fig:init_step}
\end{figure*}

\begin{figure*}[t]
\centering
\includegraphics[width=0.22\textwidth]{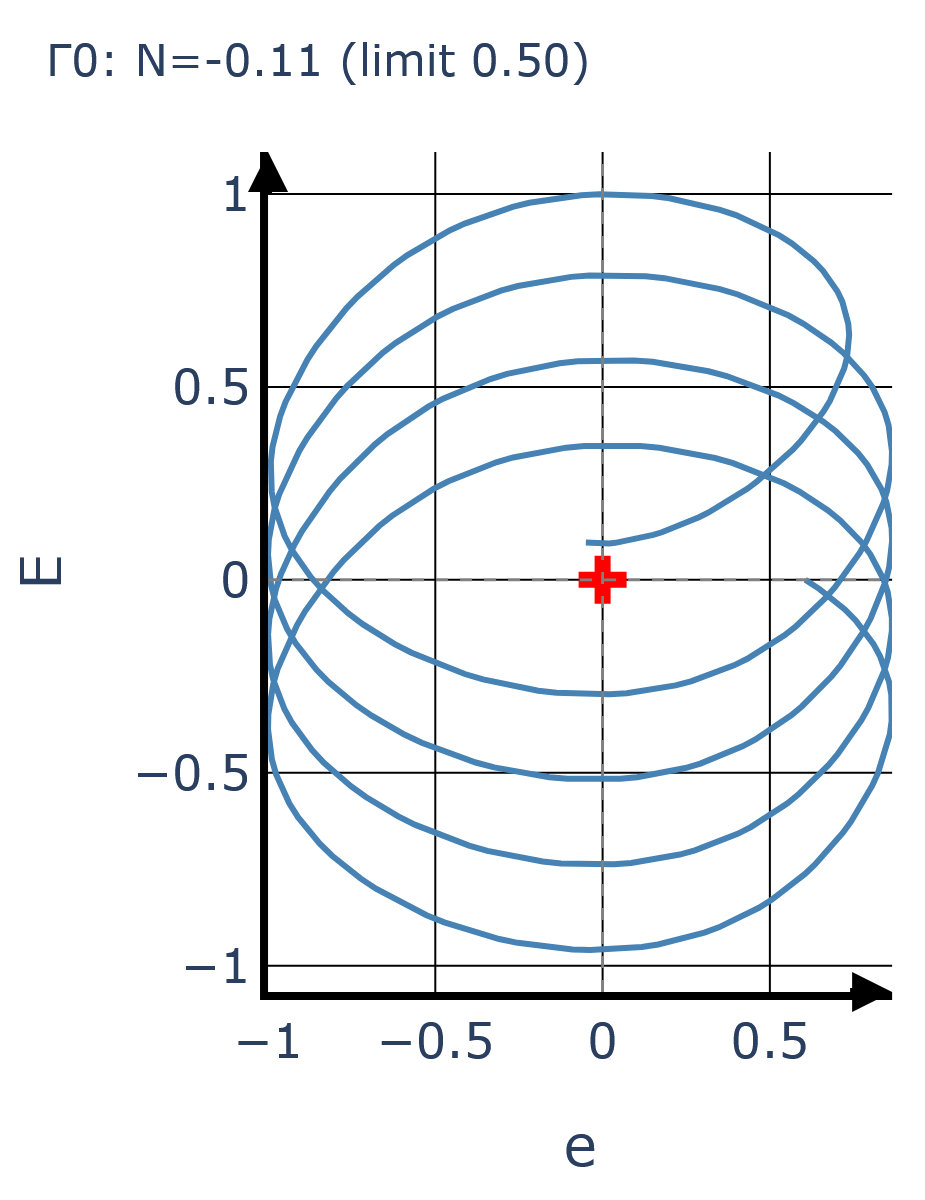}\hfill
\includegraphics[width=0.22\textwidth]{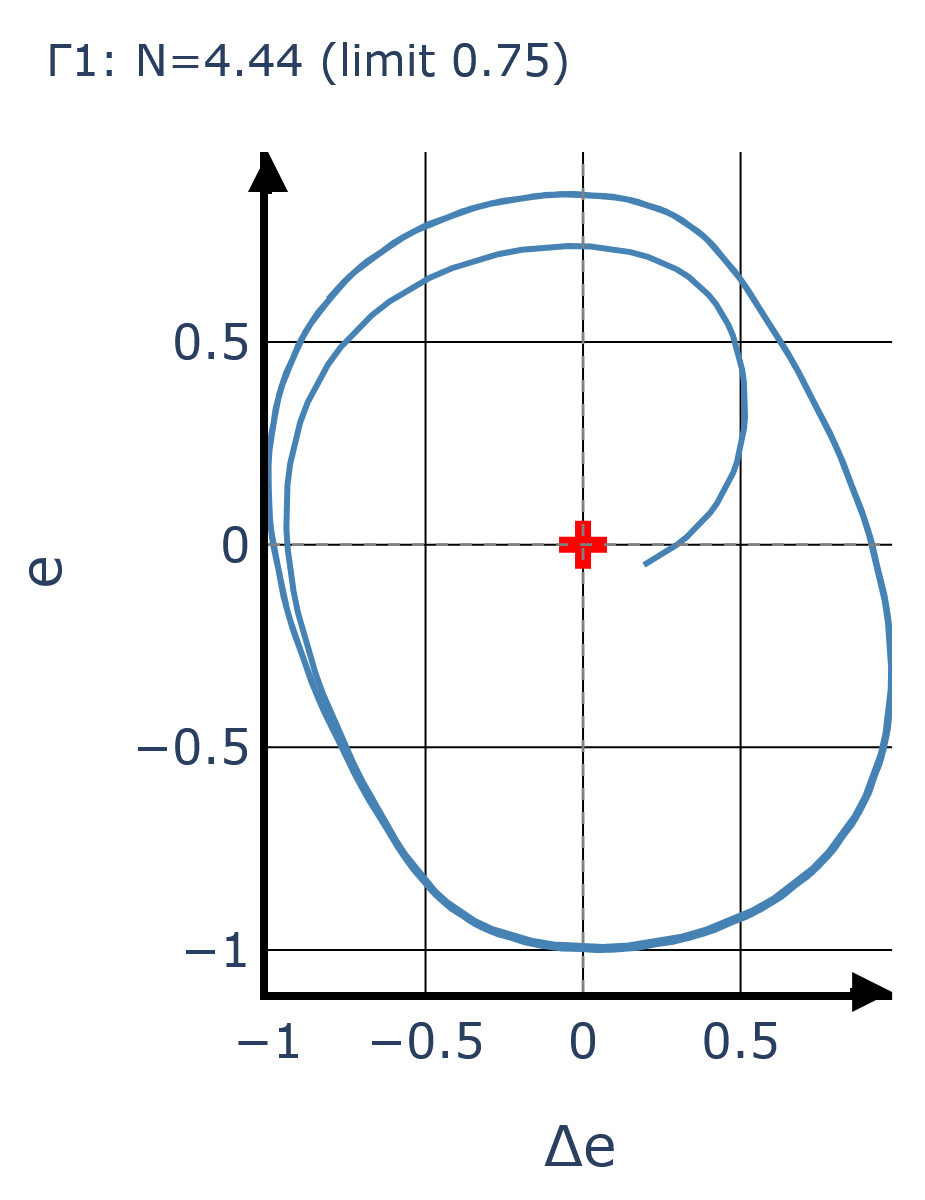}\hfill
\includegraphics[width=0.22\textwidth]{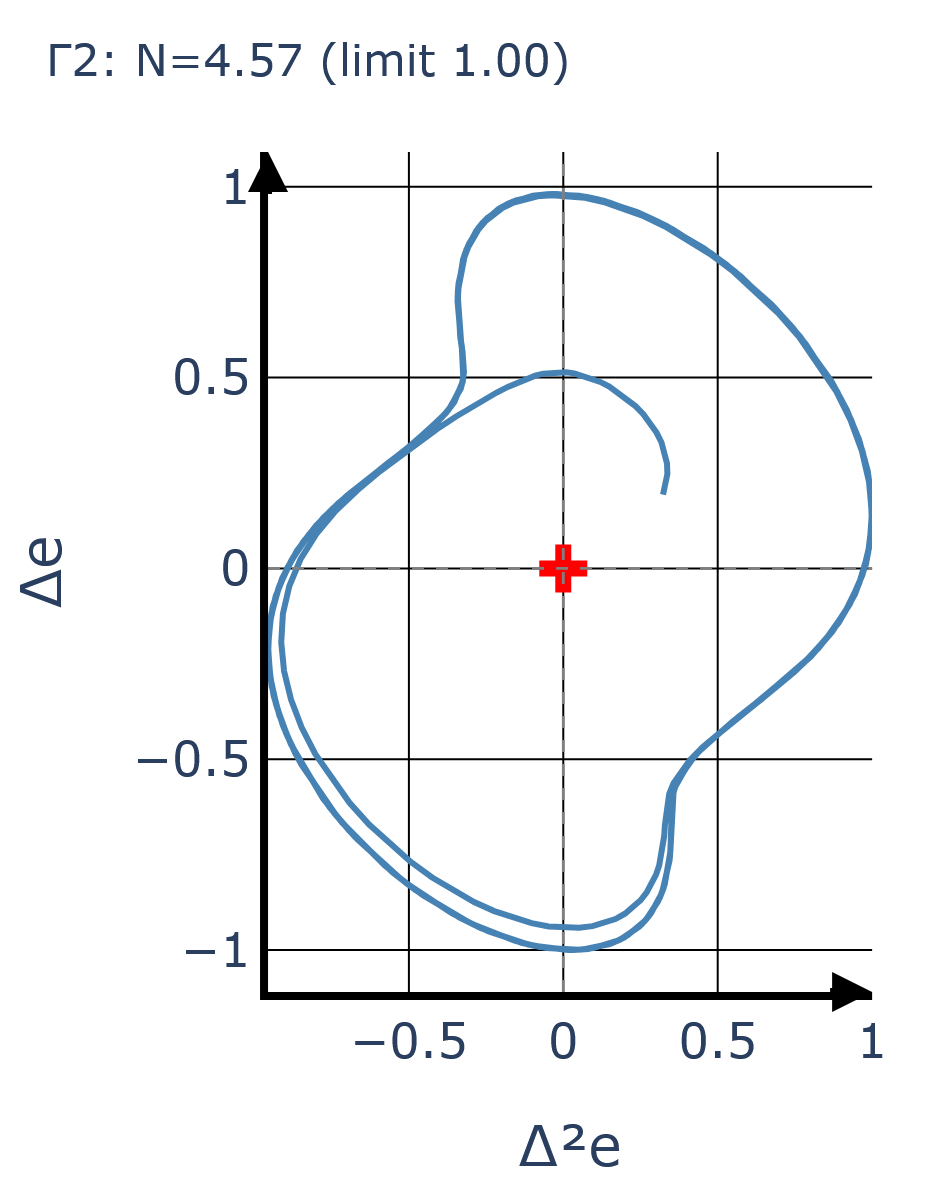}\\[2pt]
\includegraphics[width=0.22\textwidth]{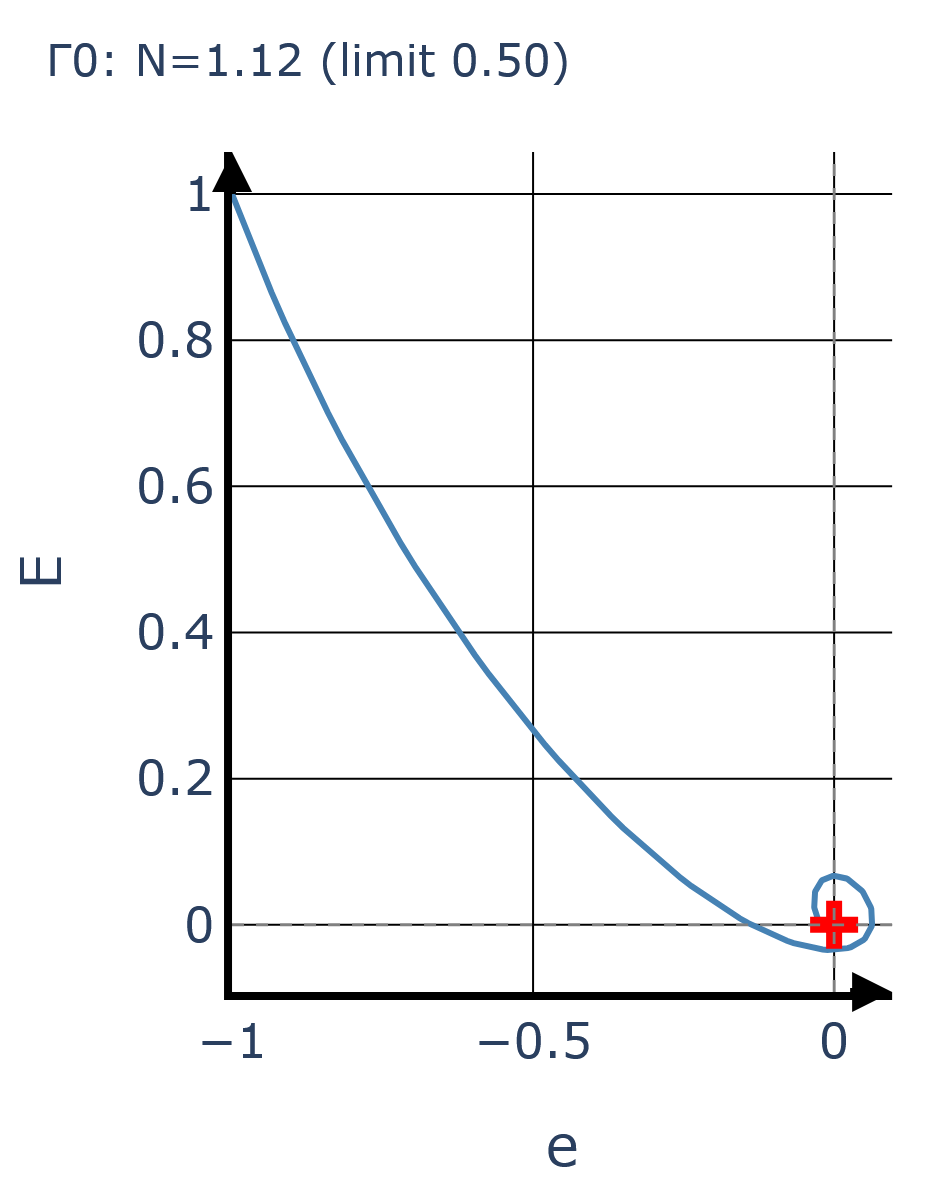}\hfill
\includegraphics[width=0.22\textwidth]{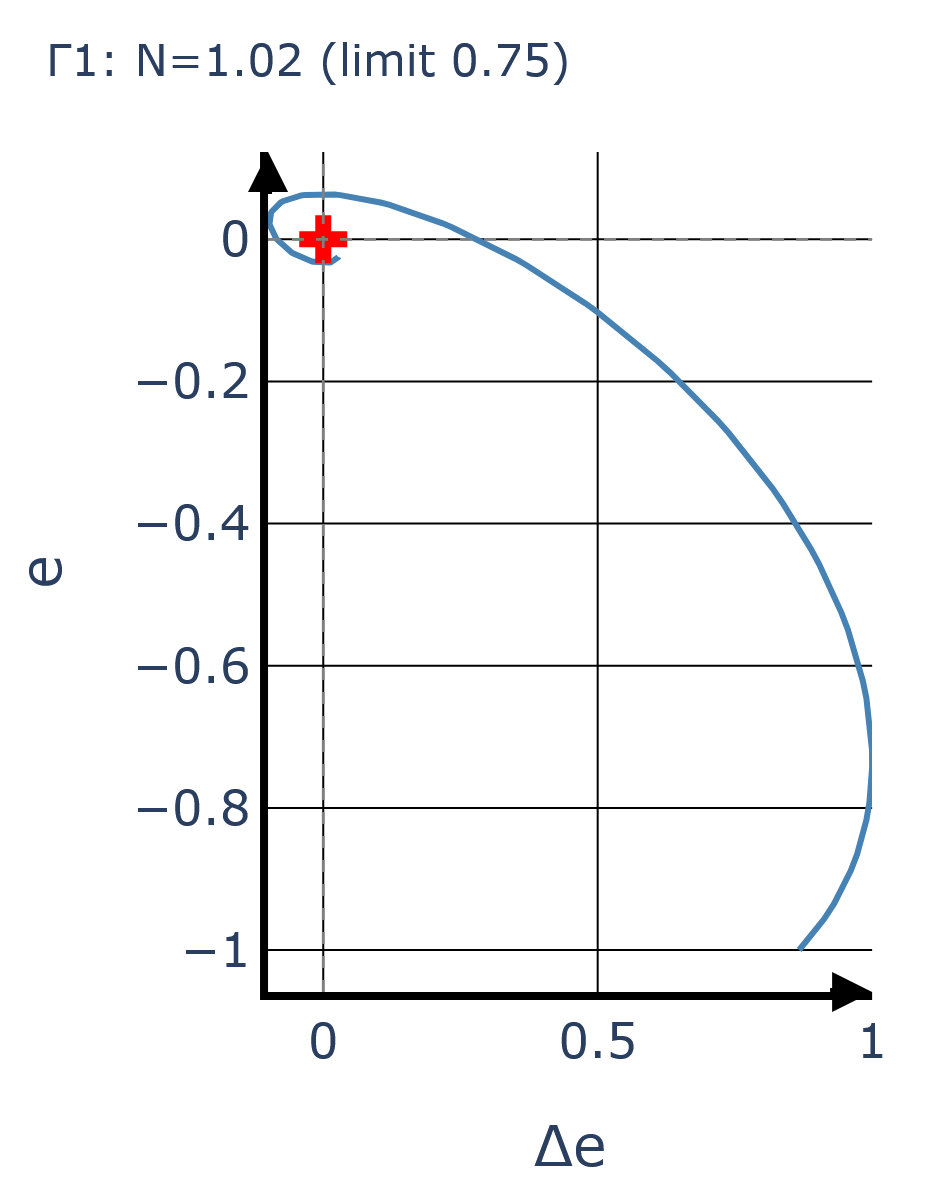}\hfill
\includegraphics[width=0.22\textwidth]{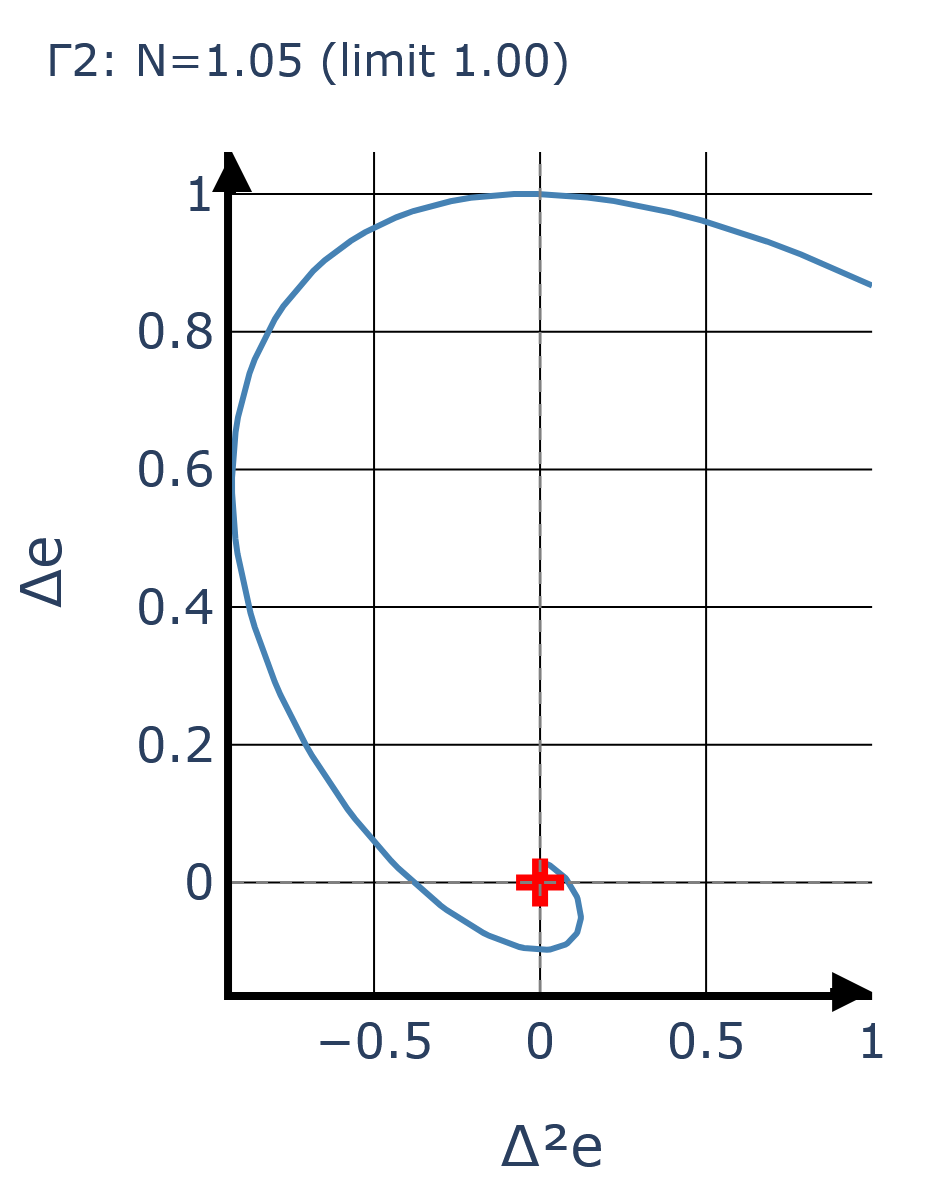}
\caption{Portraits $\Gamma_0,\Gamma_1,\Gamma_2$ (left to right) before
(top) and at the end of the run (bottom); the cross marks the settling
point. Figure~\ref{fig:init_step} does not show which gain is at fault, but
the portraits do. At unit gains the two upper bands are far
outside their limits ($N_1{=}4.44$ against $0.75$; $N_2{=}4.57$ against
$1.0$), while $N_0{=}{-}0.11$ stays inside its limit; the stability screen of
Section~\ref{sec:wellposed} rejects that record anyway. By the end of the run
the counts have fallen to $(1.12, 1.02, 1.05)$, and the last feasible
iterate is $81$ iterations earlier, at $(0.12, 0.10, 0.35)$: the
boundary limit cycle of Proposition~\ref{prop:conv}, seen from its
violated side.}
\label{fig:init_portraits}
\end{figure*}

%=====================================================================
\section{The turn index}
\label{sec:index}
%=====================================================================

\new{\spin} rests on two standing assumptions, stated once and used
throughout.

\begin{assumption}[Process class]\label{ass:plant}
\fix{The plant $G$ is SISO, asymptotically stable and non-integrating
(self-regulating). On $(0,\infty)$ both $|G|$ and $\arg G$ decrease
strictly, and $\arg G$ falls below $-3\pi/2$.}
\end{assumption}

\fix{The last clause is an arithmetic of phase. To destabilize a loop
the total phase must reach $-\pi$, and the controller of
Assumption~\ref{ass:pid} supplies at most $\pi/2$ of lead, so the plant
itself must be able to reach $-\pi-\pi/2$. This is what places the
stability boundary at a finite gain. The other clauses are used as well: asymptotic
stability lets the Nyquist argument of Proposition~\ref{prop:pair} count
encirclements; self-regulation, that is $|G(0)| < \infty$, fixes the sign
in Proposition~\ref{prop:mono}(i), so that reducing all gains damps the
loop rather than destabilizing it; and the two monotonicity clauses make
a single crossover carry the loop. Membership in this class can be
decided by inspection.}
\begin{lemma}[Lag chains with dead time]\label{lem:class}
Let $G(s) = K\,e^{-Ls}\prod_{i=1}^{n}(\tau_i s+1)^{-1}$ be a chain of
$n \ge 1$ first-order lags, with static gain $K>0$, lag time constants
$\tau_i>0$ and dead time $L\ge 0$. Then $|G|$ and $\arg G$ decrease
strictly on $(0,\infty)$, and $\arg G$ falls below $-3\pi/2$ if and only
if $L>0$, or $L=0$ and $n\ge 4$. Such a plant satisfies
Assumption~\ref{ass:plant} exactly when $L>0$ or $n\ge 4$.
\end{lemma}

\begin{proof}
\fix{Both derivatives are sums of strictly negative terms for
$\omega>0$, and $\arg G \sim -L\omega - n\pi/2$ as $\omega\to\infty$,
which diverges for $L>0$ and otherwise approaches $-n\pi/2$ from
above.}
\end{proof}

\fix{Models with dead time therefore qualify for every $L>0$, and
delay-free lag chains from four lags up. A lightly damped complex pole
pair, an integrating plant and a plant with finite zeros fall outside.
The assumption can also be read from a measured frequency response
\cite{SchlegelCech2005,Schlegel2002}.}
\new{Proposition~\ref{prop:pair} additionally needs $|CG| \to 0$, so
that the phase crossings die out, and how many lags that takes depends
on the controller rather than on Assumption~\ref{ass:plant}. For the I
and PI members $|C|$ is bounded at high frequency, so $n \ge 1$
suffices. The PID member with an ideal derivative has
$|C| \sim K_d\omega$, which cancels one lag, so it needs $n \ge 2$: with
a single lag $|CG|$ tends to the constant $K K_d/\tau_1$ and the
crossings recur without decaying.}

\begin{remark}[Plants with no stability boundary]\label{rem:ladder}
The exclusion easiest to overlook is the mildest-looking one: a
delay-free chain with $n \le 2$, and in particular $K/(\tau s + 1)$.
There $\arg G > -\pi$, the locus of $CG$ never reaches the negative real
axis, and no finite gain destabilizes the loop. The boundary of
Proposition~\ref{prop:pair} does not merely lie far away; it does not
exist. No band ever objects, the
last row of Table~\ref{tab:rule} fires at every iteration, and the gains
rise to $F_{\max}$, which Section~\ref{sec:limits} reads as the answer
it is. Any dead time, or two further lags, restores the boundary and
with it the method.
\end{remark}

\begin{assumption}[Controller structure]\label{ass:pid}
The controller is drawn from the nested family
$\mathrm{I} \subset \mathrm{PI} \subset \mathrm{PID}$. It has a
parallel structure in a unity-feedback loop. The integral channel is
mandatory ($K_i > 0$); the proportional and derivative channels are
optional, $K_p, K_d \ge 0$. \new{The derivative channel is ideal; any implementation
filters it, and Section~\ref{sec:battery} records what the reference
implementation uses.}
\end{assumption}

Write $C(s) = K_p + K_i/s + K_d s$ for the transfer function of the
controller of Assumption~\ref{ass:pid}.
Write $\theta = \new{(K_i, K_p, K_d)}$ for the controller's gains. A
uniform multiplier $m>0$ scales $\theta \mapsto m\theta$. This is an
analysis direction, and it is distinct from the per-band multipliers
$(F_i, F_p, F_d)$ used by the rule (Section~\ref{sec:rule}). \new{Each
channel is linear in its own gain, so $C$ is homogeneous of degree one
in $\theta$.} Scaling all three gains by $m$ therefore maps the loop
$CG$ to $m\,CG$ exactly, and the Nyquist locus changes in magnitude but
not in shape. \new{Homogeneity, not linearity, is all that is used
below.} A plant with an integrator
would satisfy the two monotonicity hypotheses, but it would reverse the
sign in Proposition~\ref{prop:mono}(i). Reducing all gains would then
destabilize the loop instead of damping it, which would also remove the
basis of Remark~\ref{rem:bootstrap}. The \fix{non-integrating clause of
Assumption~\ref{ass:plant}} excludes this case. \new{Its phase clause
excludes a second family with the same reversal
(Remark~\ref{rem:ladder}).}

Write $y$ for the measured output, $u$ for the control signal and $e$
for the control deviation $y-r$. A tuning experiment applies a unit
setpoint step from rest, so that $e$ runs from $-1$ to $0$, and records
$e_k$, $k = 0,\dots,M$, at period $T_s$, over a window in which the
response settles. Form the running integral and the first
and second differences,
\begin{equation}
E_k = T_s\!\sum_{j=0}^{k} e_j, \quad
\Delta e_k = e_k - e_{k-1}, \quad
\Delta^2 e_k = \Delta e_k - \Delta e_{k-1},
\end{equation}
with $\Delta e_0 = \Delta^2 e_0 = 0$. Each portrait plots one of these
signals vertically against its own difference horizontally:
\begin{align}
\Gamma_0&: \bigl(e_k,\, E_k{-}E_{M}\bigr), \notag\\
\Gamma_1&: \bigl(\Delta e_k,\, e_k\bigr), \notag\\
\Gamma_2&: \bigl(\Delta^2 e_k,\, \Delta e_k\bigr).
\label{eq:plots}
\end{align}
The signal on the vertical axis is the \emph{parent} of its portrait:
$E$, $e$ and $\Delta e$ for $\Gamma_0$, $\Gamma_1$, $\Gamma_2$. Up to the exchange of axes, $\Gamma_1$ is the
classical $(e,\dot e)$ phase plane, long used to read nonlinear servo
transients graphically \cite{GrahamMcRuer1961}, and $\Gamma_0$ and
$\Gamma_2$ are its integrated and differentiated versions. All three end
at the origin when the loop settles. Putting the difference on the
horizontal axis makes an under-damped response wind counter-clockwise in
all three. We call
\eqref{eq:plots} \emph{Pachner portraits}, after the first author, who
devised the construction.

A well-damped loop spirals into the origin without completing a full
turn. An under-damped loop encircles it. The following definition makes
the count precise.

\begin{definition}[Turn index]\label{def:index}
Given a planar trajectory $(p_k, q_k)$, $k = 0,\dots,J$, with neither
coordinate identically zero:
(i) normalize each coordinate by its peak magnitude, mapping the curve
into $[-1,1]^2$;
(ii) after the trajectory first leaves the disc of radius
$\varepsilon \in (0,1)$ about the origin, truncate at the last sample
lying inside it---the last, not the first, so that a record
still winding after it settles is not silently discarded
(Section~\ref{sec:wellposed})\fix{, and retaining the curve in full if it
does not re-enter};
(iii) define $N$ as the winding number of the remaining curve about the
origin, that is, the total swept angle divided by $2\pi$. We evaluate it
as the difference of the endpoint angles, corrected by the signed
crossings of the negative horizontal semi-axis. Partial revolutions
count fractionally.
\rev{After the normalization each coordinate attains magnitude one, so
the first exit in step~(ii) always occurs. For a trajectory with a
coordinate identically zero we set $N := 0$.}
\end{definition}

Because $N$ is the winding number about the point at which the loop
settles, we call it the \emph{turn index}.

Definition~\ref{def:index} assumes the trajectory ends at the origin,
which the integral channel guarantees since $e \to 0$ for every stable
tuning of Assumption~\ref{ass:pid}. Proportional action alone, on a
self-regulating plant, would settle at an offset: $\Gamma_0$ and
$\Gamma_1$ would end away from the origin, the truncation would never
occur, and the winding number would measure the wrong quantity. This is
what $K_i > 0$ means in practice.

Applying Definition~\ref{def:index} to $\Gamma_0, \Gamma_1, \Gamma_2$
gives the turn-index vector $\mathbf N(\theta) = (N_0, N_1, N_2)$ at the
controller setting $\theta = \new{(K_i, K_p, K_d)}$. It comes from a
single step test, and it uses no knowledge of the plant.

Four propositions carry the construction. First, the count is free of
scale and of clock (Proposition~\ref{prop:invariance}). Second, one
damping dominates the record, because stability is lost through one
critical pole pair (Proposition~\ref{prop:pair}). Third, the count
follows that damping, because near the boundary the trajectory is a
logarithmic spiral (Proposition~\ref{prop:law}). Fourth, each portrait
isolates one frequency band, because each differencing tilts the
spectrum by one power of $\omega$ (Proposition~\ref{prop:band}).

\begin{proposition}[Topological invariance]\label{prop:invariance}
$\mathbf N$ is invariant under time reparametrization $t \mapsto a t$,
$a > 0$, and under independent positive scaling of either coordinate of
each trajectory, in particular under scaling of the setpoint step. It is
also unchanged by $e \mapsto -e$, which rotates every portrait by $\pi$:
the opposite sign convention gives the same counts.
\end{proposition}

\begin{proof}
The winding number depends only on the oriented image of the curve, not
its parametrization; per-axis peak normalization removes amplitude and
units, and the truncation disc is centred at the origin, so it is
rotation invariant.
\end{proof}

For a sampled record the time-scaling statement is exact when $T_s$ is
scaled with the time axis, and holds at fixed $T_s$ in the
fast-sampling limit.

\begin{sloppypar}
One set of limits therefore serves all loops. The defaults
$\bar{\mathbf N} = (0.5, 0.75, 1.0)$ and $\varepsilon = 0.1$ are used
unchanged in every result of this paper. They act as the definition of
``well-damped'' and not as adjustable settings. Since $N$ measures
turns, the three limits are $2$, $3$ and $4$ quarter turns. Each band is
thus one quarter turn more tolerant than the band below it.
\end{sloppypar}

Which damping the counts read, and why a single damping dominates,
follows from the plant class.

\begin{proposition}[Critical pair]\label{prop:pair}
Let $G$ satisfy Assumption~\ref{ass:plant} and scale all gains
together, $CG \mapsto m\,CG$ with $m>0$. Then the loop is stable up to a
finite $m^\ast$, and at $m^\ast$ it loses stability through one complex
pole pair. This happens at the phase-crossover frequency
$\omega_{180}$, which maximizes $|CG(j\omega)|$ over
$\{\omega:\arg CG(j\omega)\rev{\equiv-\pi \pmod{2\pi}}\}$. Near $m^\ast$
this pair is the least damped in the loop, and its damping $\zeta(m)$
decreases continuously to zero as $m\uparrow m^\ast$.
\end{proposition}

\begin{proof}
Since $\Re\,C(j\omega) = K_p \ge 0$, the phase $\arg C$ lies in
$[-\pi/2,\pi/2]$, while \fix{$\arg G$ falls below $-3\pi/2$
(Assumption~\ref{ass:plant})}. Hence $\arg CG$ reaches $-\pi$ at
frequencies bounded away from $0$ and from $\infty$. Increasing $m$
scales the locus $m\,CG$ radially. By the Nyquist criterion
\cite{Nyquist1932} it first touches $-1$ at
$m^\ast = 1/|CG(j\omega_{180})|$, and generically at one frequency. The
touch is transversal, because $\partial|m\,CG|/\partial m = |CG| > 0$.
One conjugate pair therefore crosses the imaginary axis at
$\omega_{180}$, while all other poles keep strictly positive damping
nearby. \rev{The transversal crossing gives the strict decrease of
$\zeta(m)$ near $m^\ast$, and continuity gives $\zeta(m) \to 0$.}
\end{proof}

\rev{The modulo formulation is needed because dead time makes
$\arg CG$ pass $-\pi$ on infinitely many branches; the maximum over them
is attained because $|CG| \to 0$.}

Near the boundary a single pole pair therefore carries the ringing, and
we write
\begin{equation}
e(t) = x(t) + r(t), \quad
x(t) = A\,e^{-\alpha t}\cos(\omega t+\varphi),
\label{eq:mode}
\end{equation}
where $x$ is the critical mode, with $\omega_n=\sqrt{\alpha^2+\omega^2}$
and damping $\zeta = \alpha/\omega_n$, and $r$ is the remaining,
faster-settling part of the error. \fix{The mode comes to dominate the record as
$m \uparrow m^\ast$, because its damping tends to zero while $r$ keeps
its decay rate. The decomposition is asymptotic in this sense rather
than exact.} The rule operates between the well-damped boundary and the
stability boundary, and there \eqref{eq:mode} serves as the working
model; Section~\ref{sec:battery} validates it. The propositions below
are based on \eqref{eq:mode}.

\begin{figure*}[t]
\centering
\includegraphics[width=0.64\textwidth]{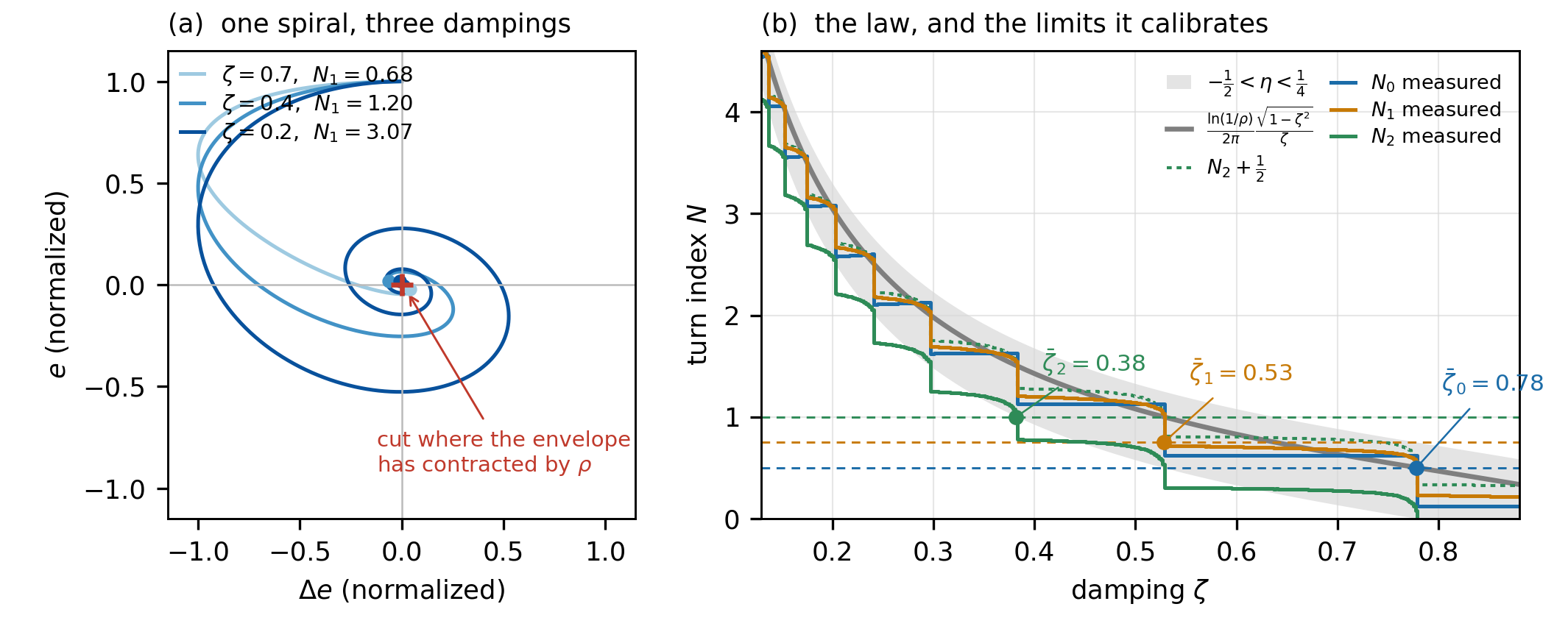}
\caption{Proposition~\ref{prop:law}. (a)~One logarithmic spiral at three dampings,
here through $\Gamma_1$; the other two are linear images of it. Each curve is cut where the envelope has contracted by $\rho$; the lower
the damping, the more turns that takes. (b)~The turn index against damping for the exact
second-order step error, $\rho = \delta = 0.02$. The leading term of \eqref{eq:law} is smooth,
while the measured counts are piecewise constant, because the cut moves
by half a cycle at a time. The shaded band is
$-\tfrac12 < \eta < \tfrac14$, the half cycle of the cut below and the
start term above. $\Gamma_2$ runs half a turn low because
$\Delta^2 e_0 = 0$ starts it at the origin; the dotted curve is
$N_2 + \tfrac12$, which lies in the same band as the other two. Each
limit is met on a step, whose midpoint is $\bar\zeta_k$.}
\label{fig:law}
\end{figure*}

\begin{proposition}[Damping law]\label{prop:law}
Let the mode of \eqref{eq:mode} be recorded until its envelope has
contracted by a factor $\rho \in (0,1)$. Then every portrait of
Definition~\ref{def:index}, truncated there, winds
\begin{equation}
N = \frac{\ln(1/\rho)}{2\pi}\,
\frac{\sqrt{1-\zeta^2}}{\zeta} + \eta,
\qquad |\eta| < 1,
\label{eq:law}
\end{equation}
where $\eta$ collects the phase of the cut, the phase of the start and
the linear distortion of the angle. The leading term is strictly
decreasing on $(0,1)$.
\end{proposition}

\begin{proof}
In the mode's own phase plane $(x, \dot x/\omega_n)$ the trajectory is a
circular logarithmic spiral, whose angle advances at the rate $\omega$
while the envelope decays at $\alpha$. Contraction by $\rho$ takes the
time $\alpha^{-1}\ln(1/\rho)$, over which the angle advances by
$(\ln(1/\rho)/2\pi)(\omega/\alpha)$ turns, with
$\omega/\alpha = \sqrt{1-\zeta^2}/\zeta$. Three terms perturb this count. The cut is the last crossing of the
settling level, at most half a cycle before the envelope reaches it, and
half rather than a whole because $|e|$ peaks twice per cycle. The record
starts where the envelope still exceeds the first peak of $|e|$,
returning $\arcsin\zeta/2\pi<1/4$ of a turn. And differentiation
multiplies the mode by $-\alpha+j\omega$, so every portrait is a linear
image of that plane, as is the peak normalization; a linear isomorphism
induces a monotone degree-one map of directions and shifts an accumulated
angle by less than $\pi$. Monotonicity follows from
$\mathrm d(\sqrt{1-\zeta^2}/\zeta)/\mathrm d\zeta<0$.
\end{proof}

\rev{The stopping rule fixes $\rho$: cutting at the first meeting with
the $\varepsilon$-disc gives $\rho = \varepsilon$, while under the
settling-anchored window of Section~\ref{sec:wellposed} the record ends
first whenever $\delta < \varepsilon$, and then $\rho = \delta$, which is
the case for the defaults of Table~\ref{tab:defaults}. Measured over $\zeta \in [0.13, 0.90]$, the residual
holds to $-\tfrac12 < \eta < \tfrac14$ for $\Gamma_0$ and $\Gamma_1$, and
to that interval shifted down half a turn for $\Gamma_2$, which starts at
the origin since $\Delta^2 e_0 = 0$ (Figure~\ref{fig:law}). Each $N_k$ is
piecewise constant with jumps of at most half a turn, the step of the
cut, so each crossing $N_k = \bar N_k$ fixes an interval of damping whose
midpoint is the calibrated $\bar\zeta_k$.}

It follows that\rev{, up to the endpoint band,}
$N_k>\bar N_k \iff \zeta<\bar\zeta_k$: the verdict ``rings'' is a sign
test, free of scale and of clock by
Proposition~\ref{prop:invariance}. The rule uses only this sign. Since $\eta$ shifts the constant in
\eqref{eq:law} without changing its monotonicity, the damping
specifications are \emph{calibrated} rather than derived. For the exact second-order step error the measured $N_k$ cross
$\bar{\mathbf N}=(0.5,0.75,1.0)$ at
$\bar\zeta \approx 0.78,\,0.53,\,0.38$ (Figure~\ref{fig:law}b), so the
specification is strictest in the lowest band.

\begin{proposition}[Band selectivity and upward leakage]\label{prop:band}
Write $e = x + r$ as in \eqref{eq:mode}\rev{, and let the residual $r$
be completely monotone with rates bounded by
$\alpha_r<\omega_n$, that is,
$r(t)=\int_{(0,\alpha_r]}e^{-\lambda t}\,\mathrm d\mu(\lambda)$ for a
finite nonnegative measure $\mu$. Write $w_j$ for the visibility of the mode in the parent signal of
$\Gamma_j$, the ratio of its amplitude to the magnitude of the residual
there. Then
\begin{equation}
w_j/w_{j-1} \;\ge\; \omega_n/\alpha_r \;>\; 1
\label{eq:leak}
\end{equation}
at every time and for every $j$. The mode is therefore suppressed in
every band below its own, and its visibility increases from its own band
upward. \fix{Its own band is therefore the lowest one whose count it
lifts, and} the set of bands whose count a mode lifts is upward-closed,
$\{k,k{+}1,\dots\}$}.
\end{proposition}

\rev{\begin{proof}
Differentiation multiplies the amplitude of the mode by exactly
$\omega_n$, since
$\tfrac{\mathrm d}{\mathrm dt}\,e^{-\alpha t}\cos(\omega t+\varphi)
=\omega_n\,e^{-\alpha t}\cos(\omega t+\varphi+\psi)$ with
$\psi=\arg(-\alpha+j\omega)$, and integration multiplies it by exactly
$1/\omega_n$. For the residual, differentiating under the integral gives
$r^{(j)}(t)=(-1)^j\!\int \lambda^j e^{-\lambda t}\,\mathrm d\mu(\lambda)$,
so that $|r^{(j)}(t)|\le\alpha_r^{\,j}r(t)$ at every $t$, while the
$j$-fold antiderivative
$\int \lambda^{-j}e^{-\lambda t}\,\mathrm d\mu(\lambda)$ is at least
$\alpha_r^{-j}r(t)$. One band step therefore scales the mode by
$\omega_n^{\pm1}$ and the residual by at most $\alpha_r^{\pm1}$ in the
same direction, which is \eqref{eq:leak}. The bound holds pointwise in
$t$, so it survives the peak normalization of
Definition~\ref{def:index}. Upward closure follows because a count is nondecreasing in the
visibility of the mode that drives it. Several modes superpose, each
contributing an upward-closed set; interference between them is
neglected, and the synthetic tests below and Section~\ref{sec:battery}
support the approximation.
\end{proof}}

Proposition~\ref{prop:band} orders the three portraits by frequency, but
it does not yet attach any of them to a controller channel. What
supplies that link is the way the loop gain responds to each gain
separately.

\begin{lemma}[Channel ownership]\label{lem:own}
Let $C$ be the controller of Assumption~\ref{ass:pid}, and let
$\sigma_k(\omega) = \partial \log|C|/\partial \log K^{(k)}$ be the
elasticity of the loop gain with respect to the gain of channel $k$.
Then
\begin{equation}
\sigma_0(\omega) + \sigma_1(\omega) + \sigma_2(\omega) = 1
\label{eq:partition}
\end{equation}
at every $\omega$, and each term tends to $1$ on its own band:
$\sigma_0 \to 1$ as $\omega \to 0$, $\sigma_2 \to 1$ as
$\omega \to \infty$, and $\sigma_1 \to 1$ wherever $K_p$ dominates $|C|$.
\end{lemma}

\begin{proof}
Write $C(j\omega) = K_p + jX$ with $X = K_d\omega - K_i/\omega$, so that
$|C|^2 = K_p^2 + X^2$. Differentiating gives
$\sigma_1 = K_p^2/|C|^2$, $\sigma_2 = X K_d\omega/|C|^2$ and
$\sigma_0 = -X K_i/(\omega|C|^2)$, and their sum is
$(K_p^2 + X^2)/|C|^2 = 1$. This is Euler's identity for the degree-one
homogeneity of $C$ in $\theta$. The limits follow because exactly one of
$K_i/\omega$, $K_p$ and $K_d\omega$ dominates $|C|$ on each band.
\end{proof}

Since $G$ does not depend on the gains, $\sigma_k$ is also the elasticity
of $|CG|$. Changing $K^{(k)}$ therefore moves the loop gain essentially
only on band $k$, and by Proposition~\ref{prop:pair} it is the damping of
the crossover in that band which responds, while
Proposition~\ref{prop:band} says $\Gamma_k$ is the portrait that reads
the same band. The two orderings agree: this is the sense in which $N_0$
points to $K_i$, $N_1$ to $K_p$ and $N_2$ to $K_d$, and that
correspondence is the entire basis of the tuning rule.

\fix{Two qualifications belong with \eqref{eq:partition}. It partitions
unity but not into nonnegative parts, since $\sigma_0$ and $\sigma_2$ dip
slightly below zero where the two reactances cancel. And the bands must
exist: where the integral and derivative times $T_i = K_p/K_i$ and
$T_d = K_d/K_p$ are close, no frequency range is owned by $K_p$ and
$N_1$ loses its selectivity.}

\emph{Attribution.} The violated set is a union of upward-closed sets
(Proposition~\ref{prop:band}). Its lowest member is therefore the only
band whose violation cannot be leakage, because there is no band below
it from which leakage could come. Violations above it carry no evidence
either way. Section~\ref{sec:battery} shows the asymmetry on a real
plant, and the rule of Section~\ref{sec:rule} is built on it.

%=====================================================================
\section{Well-posedness of the count}
\label{sec:wellposed}
%=====================================================================

Normalizing each axis separately gives
Proposition~\ref{prop:invariance}, but removes all information about
absolute scale. A small persistent oscillation, from quantization, a weakly damped
parasitic mode or numerical residue, can then dominate $\Delta e$ and
$\Delta^2 e$ long after $e$ has settled: normalized, it fills the
portrait and winds without end, and the comparison with the limits
becomes an artifact of the window.

This failure is not hypothetical. Take $G(s) = e^{-4s}/(8s+1)^6$ at
gains the rule reaches during an iteration. The response settles well
inside the default window, yet lengthening that window sixfold raises
$N_1$ from $4.13$ to $7.89$ and $N_2$ from $4.31$ to $17.9$. Meanwhile
$N_0$ does not move at all, since integration suppresses the fast
residual that differencing amplifies
(Proposition~\ref{prop:band}). The growth stops only where the residue
reaches the resolution of the record; in a noisy installation the residue
does not decay, and it would not stop at all. In closed loop the comparisons with the limits then follow the window
rather than the loop, and the corrective moves work against each other
rather than converging.

The remedy is to anchor the window to the settling of the parent signal.

\begin{definition}[Settling-anchored window]\label{def:guard}
Fix $\delta \in (0,1)$ (default $\delta = 0.02$). Let
$k_\delta = 1 + \max\{k : |e_k| > \delta \max_j |e_j|\}$ and truncate all
three trajectories \eqref{eq:plots} at $k_\delta$ before applying
Definition~\ref{def:index}.
\end{definition}

\begin{proposition}[Well-posedness]\label{prop:wellposed}
Under Definition~\ref{def:guard}, the turn indices $\mathbf N$ do not
depend on the window length, for every window that contains the
$\delta$-settling time of $e$. Proposition~\ref{prop:invariance}
continues to hold, because $\delta$ has no units and is defined on the
normalized error.
\end{proposition}

\fix{Both claims are immediate: every sample beyond $k_\delta$ is
discarded, $k_\delta$ depends only on $e$ up to its $\delta$-settling
time, and $k_\delta$ is defined through the ratio
$|e_k|/\max_j|e_j|$.}

On the counterexample the guarded counts hold to three digits on all
three windows, at $N_1 = 2.04$ and $N_2 = 2.11$, and
Section~\ref{sec:battery} shows the closed-loop iteration on the same
plant converging cleanly. From here on the guard is part of the index. The failure it excludes is
silent and window-dependent, so it deserves a formal guard rather than
the judgment of the operator. The definition then carries exactly the
dimensionless defaults $(\bar{\mathbf N}, \varepsilon, \delta)$. The rule of
Section~\ref{sec:rule} adds the step $\beta$ and the multiplier box,
which completes Table~\ref{tab:defaults}. None of these defaults was
adjusted for any plant in this paper.

\begin{remark}[Scope of the guard]\label{rem:noise}
Proposition~\ref{prop:wellposed} assumes $e$ is noise-free. With noise
level $\nu$ the index $k_\delta$ is well defined only if
$\delta \max_j |e_j| \gg \nu$; otherwise it follows the noise.
The default $\delta = 0.02$ therefore assumes noise below roughly two
percent of the step, which should be checked and not taken for granted.
In practice, keep tens to a few hundred samples per transient so that
$\Delta e$ and $\Delta^2 e$ are not dominated by quantization, and
difference only from the step onward. \new{Filtering the derivative
channel does not help: the portraits \eqref{eq:plots} difference the
recorded error outside the controller, so the noise reaches them
anyway.}
\end{remark}

\emph{A stability screen.} The guard, and the counts themselves, assume
that the record settles at all; the winding of a diverging error is
merely an artifact of where the recording stopped. Whether it settles
can be decided from the record itself, with no model and no
thresholds.

\begin{definition}[Stability screen]\label{def:screen}
Model the recorded error as the output of two zero-mean Gaussian
generators switched once at an unknown sample, a maximum-likelihood
changepoint test \cite{BassevilleNikiforov1993}. The switch point is
\begin{equation}
\ell^\ast = \arg\max_{\ell}\;
\bigl[-\ell \ln s_1^2(\ell) - \rev{(M{-}\ell{+}1)}\ln s_2^2(\ell)\bigr],
\end{equation}
where $s_1^2$ and $s_2^2$ are the mean squares of $e_0,\dots,e_{\ell-1}$
and of $e_\ell,\dots,e_M$, and $\ell$ ranges over splits leaving at
least $10\%$ of the record on each side. \fix{For a
margin $c > 1$ (default $c = 2$)} the record is \emph{unstable} when
$s_2^2(\ell^\ast) \ge \fix{c\,}s_1^2(\ell^\ast)$, and \emph{stable}
otherwise.
\end{definition}

A settling response has most of its energy early and a diverging one has
most of it late, and the verdict survives scaling of either axis.
Moments are taken about zero, which $K_i > 0$ justifies. A bounded sustained oscillation
\emph{passes} the screen, and this is deliberate: a limit cycle has a
frequency and belongs to the band rules, whereas the screen is reserved
for divergence, which no count can describe. \fix{The margin secures
this, since a constant-amplitude oscillation gives
$s_2^2 \approx s_1^2$ at every split while divergence sends the ratio up
without bound.}

%=====================================================================
\section{The tuning rule}
\label{sec:rule}
%=====================================================================

Order the gains by the frequency band each channel
dominates and write
\begin{equation}
K^{(0)} = K_i, \qquad K^{(1)} = K_p, \qquad K^{(2)} = K_d,
\end{equation}
so that the band number matches the portraits: $N_k$ is the turn index
that points to $K^{(k)}$, the gain that owns band $k$
(Lemma~\ref{lem:own}). Parametrize the gains as multipliers
$F^{(k)} \in [F_{\min}, F_{\max}]$ on any conservative stabilizing
initialization, so that $F^{(0)} = F_i$, $F^{(1)} = F_p$ and
$F^{(2)} = F_d$ in the notation of Table~\ref{tab:rule};
fix a step $\beta$ (Table~\ref{tab:defaults}), and let
$\gamma = 1/(1-\beta)$. Each iteration performs one step test,
evaluates $\mathbf N$, and updates the gains.

Table~\ref{tab:rule} contains the entire decision logic. Manual
commissioning practice can be recognized row by row. The
lower-triangular pattern that gives the rule its name is visible in the
arrows.

\begin{table}[tb]
\centering
\caption{The tuning rule in compact form. The rows are read from top to
bottom, and the first matching row applies, once per step test. Here
``unstable'' is the verdict of Definition~\ref{def:screen}, and
``loops'' means $N_k > \bar N_k$.
$\uparrow$: multiply by $\gamma$;\ $\downarrow$: divide by $\gamma$;\
$\Downarrow^{a}$: divide by $2^{a}$;\
$\cdot$\,: unchanged. Projections onto $[F_{\min},F_{\max}]$ implied.
In the last column, \emph{reset} is the traditional name for integral
action and \emph{rate} for derivative action.}
\label{tab:rule}
\footnotesize
\setlength{\tabcolsep}{4pt}
\begin{tabular}{lcccl}
\toprule
Condition & $F_i$ & $F_p$ & $F_d$ & manual reading \\
\midrule
unstable & $\Downarrow^{1}$ & $\Downarrow^{2}$ & $\Downarrow^{3}$ &
diverging: back off, rate hardest \\
$\Gamma_0$ loops & $\downarrow$ & $\cdot$ & $\cdot$ &
slow cycling: less reset \\
$\Gamma_1$ loops & $\uparrow$ & $\downarrow$ & $\cdot$ &
ringing: less gain \\
$\Gamma_2$ loops & $\uparrow$ & $\uparrow$ & $\downarrow$ &
fast cycling: less rate \\
none & $\uparrow$ & $\uparrow$ & $\uparrow$ &
all quiet: tighten \\
\bottomrule
\end{tabular}
\end{table}

\begin{table}[tb]
\centering
\caption{Constants and defaults, used unchanged on every plant in this
paper.}
\label{tab:defaults}
\begin{tabular}{lll}
\toprule
Symbol & Meaning & Default \\
\midrule
$\bar{\mathbf N}$ & band limits & $(0.5, 0.75, 1.0)$ \\
$\varepsilon$ & truncation radius & $0.1$ \\
$\delta$ & settling band & $0.02$ \\
$\beta$ & step size & $0.1$ \\
\fix{$c$} & \fix{screen margin} & \fix{$2$} \\
\new{$\mathbf a$} & \new{backoff profile} & \new{$(1,2,3)$} \\
box & multiplier range & $[0.001, 10]$ \\
\bottomrule
\end{tabular}
\end{table}

A row that cuts one gain raises the gains below it. This is not
compensation: a band below the fault has just been measured within its
limit in the same experiment, so its margin is evidence, and the
objective requires spending it.

The screen verdict takes precedence. An unstable record backs the
multipliers off, and no counts are evaluated, because the portraits of
a diverging error carry no usable information. The backoff is
deliberately coarse: it halves the multipliers instead of dividing them
by $\gamma$. Instability is a state to be left quickly, not corrected
carefully. The backoff is also \new{graded}: with
$\mathbf a = (1,2,3)$ the multipliers are divided by $2$, $4$ and $8$.

\new{The grading is what keeps the rule from cycling. A uniform backoff
is the exact mirror of the last row of Table~\ref{tab:rule}, and the two
moves cancel: a sluggish loop is expanded until it goes unstable, backed
off uniformly until it is sluggish again, and returned to the same
place, since the ratios $K_i \!:\! K_p \!:\! K_d$ never move. Expansion
is the right response to no evidence, so it is the backoff that must
break the symmetry. Graded, it changes those ratios one way: each
retreat sheds derivative action for integral action until some band
objects. The order is one of cost rather than of blame, which the screen cannot
supply, and derivative action is the cheapest of the three to give up.
The spacing is arithmetic because the band edges move by the
\emph{differences} of $\mathbf a$, so equal differences slide the whole
partition one octave and leave counts comparable across a backoff. Nothing is remembered, and the retreat is never weaker than
a uniform halving in any band, so
Proposition~\ref{prop:mono}(i) continues to describe it.}
For a stable record, define the \emph{violated set} and its lowest
member as
\begin{equation}
V = \{k : N_k > \bar N_k\}, \qquad k_{\min} = \min V,
\end{equation}
with $k_{\min} = 3$ by convention if $V = \emptyset$. The update reads
\begin{equation}
F^{(k)} \;\leftarrow\;
\begin{cases}
\gamma\, F^{(k)}, & k < k_{\min} \ \text{(raise band)},\\[1mm]
F^{(k)}/\gamma, & k = k_{\min} \ \text{(cut band)},\\[1mm]
F^{(k)}, & k > k_{\min} \ \text{(untouched)}.
\end{cases}
\label{eq:triangular}
\end{equation}
When $k_{\min} = 3$ this degenerates to raising all bands.

The rule pursues one objective: the largest gains whose counts stay
within the limits. The tuning should be fast, not merely well damped,
and for a unit load disturbance the integrated error is $1/K_i$, so
raising
$K_i$ as far as the lowest band allows acts directly on what makes a
slow loop slow.

The rows of Table~\ref{tab:rule} are not chosen freely. Because
$k_{\min}$ is the lowest violated band, the bands fall into three states
of knowledge, and each state allows exactly one action.

\begin{itemize}[leftmargin=1.4em, itemsep=1pt, topsep=2pt]
\item $k < k_{\min}$: \emph{measured within limits} in this same
experiment, because $k_{\min}$ is the lowest violated band, and the
objective requires that this margin be used. $\uparrow$
\item $k = k_{\min}$: \emph{identified as the fault}. It is the only
member of $V$ whose violation cannot be leakage
(Proposition~\ref{prop:band}), and reducing its gain is effective
(Proposition~\ref{prop:mono}(ii)). $\downarrow$
\item $k > k_{\min}$: \emph{no evidence either way}. The count may be
genuine, or the mode of band $k_{\min}$ seen through further
differencing. Raising a violated band would worsen the fault; reducing
one within its limit would give up gain for nothing. $\cdot$
\end{itemize}

Reading these three cases gives \eqref{eq:triangular} entry by entry:
below $k_{\min}$ there is a measurement, above it there is none. The
argument fixes only the pattern of signs; the step size and the strength
of the third case remain free, and reducing the unobserved bands too
would serve the same objective at the price of the attribution
(Section~\ref{sec:battery}). The rule therefore acts in sequence, and
Proposition~\ref{prop:mono} guarantees that each stage ends.
In loop-shaping terms the gain profile is pushed up from the
low-frequency end and steps down at the first band that objects: a
lexicographic bandwidth maximizer, and a fixed-point search for the
boundary of the well-damped set
$\mathcal F = \{\theta : \mathbf N(\theta) \le \bar{\mathbf N}\}$. Gains
and multipliers differ only by the fixed initialization, so we write
$\mathcal F$ for the multipliers too.

One fact follows from \eqref{eq:triangular}. In logarithmic gain
coordinates the corrective moves are lower-triangular and unimodular, so
their compositions reach every point of the multiplicative
$\gamma$-lattice, and an unwanted compensation is cancelled by a
correction in a lower band one iteration later.

\begin{proposition}[Monotone steering]\label{prop:mono}
Let Assumption~\ref{ass:plant} hold, and consider an iterate at which
band $k$ is violated. The record is then measurably under-damped,
$\zeta < \bar\zeta_k$, so by Proposition~\ref{prop:pair} the loop lies in
the upper part of its gain range and the critical pair carries the
record. \fix{We argue on the leading term $\hat N(\zeta)$ of \eqref{eq:law},
since the measured count moves in half-turn steps and only the smooth
part is strictly monotone.} Then: (i)~along the uniform-multiplier
direction $CG \mapsto m\,CG$, \fix{$\zeta$ decreases strictly with $m$
and $\hat N$ increases strictly with $m$};
(ii)~\fix{the same holds for $\hat N_k$ in} its own band gain
$K^{(k)}$, at a crossover that its band owns. \fix{The counts
themselves are monotone in a weaker sense, because $N$ and $\hat N$
differ by the bounded term $\eta$ of Proposition~\ref{prop:law}. They
never move against $\hat N$, and they move with it whenever $\hat N$
sweeps a full unit.} The reduction in the corrective row of
\eqref{eq:triangular} therefore \fix{does not increase}
$N_{k_{\min}}$\fix{ and strictly decreases $\hat N_{k_{\min}}$}, and the
compensating increases do not cancel this effect.
\end{proposition}

\begin{proof}
Write $\partial \fix{\hat N}/\partial\log m =
(\partial \fix{\hat N}/\partial\zeta)\,(\partial\zeta/\partial\log m)$.
The first factor is negative by Proposition~\ref{prop:law}. For the
second factor, note that the phase-crossover frequencies of $m\,CG$ do
not move with $m$, because $\arg m\,CG = \arg CG$, while the gain at
each of them grows linearly in $m$. The margin to the first touch of
$-1$ therefore decreases strictly, and near the boundary the damping of
the critical pair follows it (Proposition~\ref{prop:pair}). The product
is positive. For~(ii), the same computation is carried out along the
band's own gain, at the crossover that its mode owns
(Proposition~\ref{prop:band}). The compensating increases act on lower
bands, which are subdominant in $|C(j\omega)|$ at that crossover, so one
step of the rule preserves the sign of the reduction. \fix{The term
$\eta$ is bounded and piecewise constant, so $N$ follows $\hat N$ in
unit jumps and cannot move against it. Part~(ii) is the weaker half: it
rests on the local dominance of $K^{(k)}$, since moving one gain shifts
its crossover too.}
\end{proof}

\begin{remark}[Scope]\label{rem:scope}
The corrective row fires only where a count is violated, which is
exactly the hypothesis of Proposition~\ref{prop:mono}: a violated count
is measured under-damping, and that places the loop in the upper part of
its gain range. Monotonicity is therefore not a further condition to be
checked at each iterate. It is a property of the states in which the rule
uses it, and Assumption~\ref{ass:plant} supplies it there.

Two limits remain. The statement is scalar, about one count, one loop
gain and one crossover; it does not claim that every pole damping is
monotone in every band gain. And strong derivative action can raise the
dominant damping over a middle range of gain, where the critical pair is
captured by the controller zeros. A violated count at low gain would
then fall outside the argument. None occurred on the battery of
Section~\ref{sec:battery}.
\end{remark}

\begin{proposition}[Behavior of the iteration]\label{prop:conv}
The iteration is deterministic without further conditions, and its
iterates remain inside the multiplier box. Assume in addition the
following. Assumption~\ref{ass:plant} holds, so that
Proposition~\ref{prop:mono} is available. The start $\theta^0$ is
accepted by the screen (Definition~\ref{def:screen}). No multiplier
reaches $F_{\min}$ on a violated band during the run. \rev{No condition
on the iterates is needed beyond these, because every corrective move is
taken at a state that satisfies the hypothesis of
Proposition~\ref{prop:mono} (Remark~\ref{rem:scope}).} Then the iteration reaches
\fix{$\partial\mathcal F$ to within one multiplicative step per gain} in
finitely many steps, or it saturates at the box if $\mathcal F$ contains
the box. After that it remains within a \fix{bounded} band around
$\partial\mathcal F$ and \fix{on it} executes a limit cycle.
\end{proposition}

\begin{proof}
Boundedness follows from the construction of the projections. On
$\mathcal F$ we have $V = \emptyset$, and the rule multiplies all gains
by $\gamma > 1$. This is an expanding map, so it leaves any compact
subset of the interior in finitely many steps, or it saturates at
$F_{\max}$. Once $V \ne \emptyset$, the corrective form of the rule
reduces the lowest violated index, and the decrease is strict \fix{in
$\hat N$} by Proposition~\ref{prop:mono}. \new{Screen events are finite
in number: each divides $K_i \!:\! K_p \!:\! K_d$ by
$1 \!:\! 2 \!:\! 4$, while the uniform expansion, the only other move
available on $\mathcal F$, leaves those ratios unchanged and cannot undo
the change. The ratios therefore never revisit a value, and the box
bounds how far they may travel, so after finitely many screen events
either a band objects or a multiplier reaches its bound, the detectable
outcome excluded by hypothesis. Thereafter the moves are
$\gamma^{\pm1}$ only, the iterates live on a finite multiplicative
lattice inside the box, and the strict decreases over the finitely many
visited states have a positive minimum.} The $F_{\min}$ hypothesis
ensures that the reduction is never blocked at the bound. When that
index clears, either the leaked violations above it clear with it, or a
genuine one becomes the new $k_{\min}$ and is reduced in its turn.
Feasibility is therefore restored in finitely many steps. \rev{On the
same finite lattice the deterministic iteration is eventually
periodic,} and the alternation of expansion and correction confines the
\rev{cycle} to \fix{a band bounded by the box. Its width does not
follow from this argument, because the bands below $k_{\min}$ are raised
until a lower violation appears; on the battery of
Section~\ref{sec:battery} the excursion is at most two multiplicative
steps per gain, which we report as a measurement}.
\end{proof}

\begin{remark}[Bootstrap from arbitrary gains]\label{rem:bootstrap}
With the screen row, a stabilizing initialization is a convenience and
not a precondition. Plants of Assumption~\ref{ass:plant} are open-loop
stable, so successive backoffs reach a region the screen accepts within
a few experiments, or else reach $F_{\min}$, a detectable outcome;
after that Proposition~\ref{prop:conv} applies unchanged. \new{The
grading shortens the bootstrap as well as the recovery, dividing the
loop gain at the top of the range by $8$ rather than $2$, which is
where a diverging record usually has its energy.} Each backoff costs
one experiment on an unstable loop, so a conservative start remains
good practice.
\end{remark}

\emph{Termination.} No stopping rule is needed. The expanding form is
active exactly on $\mathcal F$, so the iteration never rests; by
Proposition~\ref{prop:conv} it ends in a limit cycle around
$\partial\mathcal F$ whose controllers remain usable. It may be left
running or stopped at any iterate, and Section~\ref{sec:battery} quotes
the last feasible one. A step test aborts, with rollback, if $|e|$ exceeds a configured
multiple of the step: a site safety setting, not a constant of the
method. The gains change with the loop
closed, so the controller is realized in velocity form
\cite{AstromHagglundBook2006}, which absorbs a gain change without a
discontinuity in $u$ and without reinitializing any state. This form is
assumed throughout Section~\ref{sec:battery}. In a real installation the $K_d$ term is
best taken on the measurement rather than on the error, to avoid the
derivative impulse at a setpoint step.

%=====================================================================
\section{Validation on a plant battery}
\label{sec:battery}
%=====================================================================

\begin{figure*}[!t]
\centering
\includegraphics[width=0.52\textwidth]{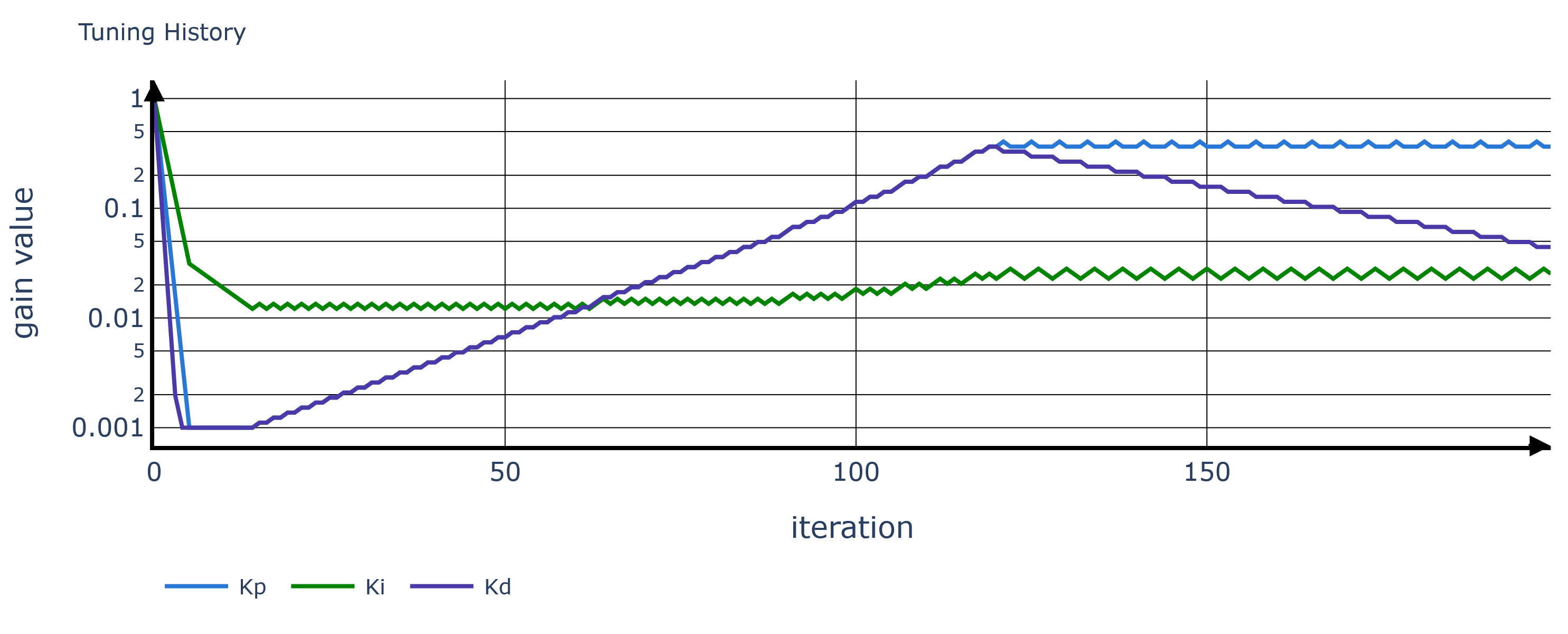}
\caption{P2, single-band rule: gain trajectories $\new{(K_i, K_p, K_d)}$
against iteration, on a log scale; the demonstration run of
Figures~\ref{fig:init_step} and~\ref{fig:init_portraits}. The start is
not stabilizing, so the screen row carries it: five graded backoffs
(iterations $1$--$5$) drive $K_p$ and $K_d$ to the lower bound $0.001$
and settle the record, and band-$0$ corrections restore feasibility by
iteration $15$ (Remark~\ref{rem:bootstrap}). $K_p$ and $K_i$ then enter
the boundary limit cycle of Proposition~\ref{prop:conv}, visible as a
ripple of one multiplicative step. From iteration $120$ the $\Gamma_2$
row reduces $K_d$ again, one step per cycle: on this delay-dominant
plant the derivative action gained during the rise adds no damping at
the boundary, so the rule gives it up.}
\label{fig:limitcycle}
\end{figure*}

\spin was run on four process-typical plants, covering lag-dominant,
balanced, delay-dominant and high-order dynamics:
\begin{align*}
\text{P1 (lag-dominant):} &\quad \tfrac{e^{-s}}{(10s+1)(s+1)^3}, \\
\text{P2 (balanced):} &\quad \tfrac{1.25\,e^{-8s}}{(5s+1)^4}, \\
\text{P3 (delay-dominant):} &\quad \tfrac{e^{-10s}}{(2s+1)(s+1)^3}, \\
\text{P4 (high-order slow):} &\quad \tfrac{e^{-4s}}{(8s+1)^6}.
\end{align*}
These four cases follow the process classification used in
\cite{AstromHagglund2004} to organize achievable PID performance.
Relative dead time is treated there as the main axis of difficulty for
a fixed PID structure. Let $L_a$ and $T_a$ be the apparent dead time and
the apparent time constant of a two-point fit of the open-loop step
response, and let $\kappa = L_a/(L_a{+}T_a)$ be the relative dead time.
Plants P1--P4 cover this axis: from lag-dominant ($\kappa \approx 0.28$,
P1) through balanced ($\kappa \approx 0.58$, P2) to strongly
delay-dominant ($\kappa \approx 0.79$, P3). Plant P4
($\kappa \approx 0.54$) instead tests a high model order at a moderate
relative dead time.

Of the dynamics Assumption~\ref{ass:plant} excludes, strong
non-minimum-phase behavior is rare in the process industries while
integrating plants are not; extending the method to them is left open
\cite{Desborough2002}. \fix{All four plants here are lag chains
with dead time and $n\ge 4$, so Lemma~\ref{lem:class} applies directly
and Assumption~\ref{ass:plant} holds} (Table~\ref{tab:battery}).

All runs use the same constants: $\bar{\mathbf N} = (0.5, 0.75, 1.0)$,
$\varepsilon = 0.1$, $\delta = 0.02$, $\beta = 0.1$, and the box
$[0.001, 10]$. The simulated plant only generates step responses; the
tuner sees the recorded error alone, exactly as it would on hardware.
\new{The reference implementation filters the derivative channel with
time constant $K_d/10$, as any implementation must; the theory above is
stated for the ideal channel of Assumption~\ref{ass:pid}.}
The derivative acts on the error, the actuator limits are $\pm 10$
under anti-windup, and each noise-free record has $500$ samples spanning
ten times the apparent time constant $T_a$. The demonstration run
of Figures~\ref{fig:init_step}--\ref{fig:limitcycle} instead takes the
derivative on the measurement, as recommended in
Section~\ref{sec:rule}; the portraits are formed from the recorded
error either way. The initial gains came from an AMIGO design
\cite{AstromHagglund2004}, computed on a two-point FOPTD fit of each
plant. The proportional and integral gains were then halved and the
derivative gain tripled: a deliberately mis-tuned start that exercises
the rule at every band position. Each run spans $200$ iterations.

\begin{table}[t]
\centering\footnotesize
\caption{Plant battery: initial and returned turn indices, and returned
multipliers, over runs of $200$ iterations. The returned tuning is the
last feasible iterate, reached at iterations $198$, $193$, $192$ and
$193$ respectively. The constants are identical throughout and the rule
is \eqref{eq:triangular}. No run fired the screen.}
\label{tab:battery}
\setlength{\tabcolsep}{3.5pt}
\begin{tabular}{@{}llll@{}}
\toprule
Plant & $\mathbf N$ initial & $\mathbf N$ returned & $(F_i, F_p, F_d)$ \\
\midrule
P1 &
$(0.62, 0.42, 1.42)$ & $(0.11, 0.02, -0.13)$ & $(2.09, 2.87, 0.53)$ \\
P2 &
$(0.12, -0.03, 2.06)$ & $(0.12, 0.07, 0.51)$ & $(2.32, 2.87, 0.59)$ \\
P3 &
$(0.16, -0.14, 7.93)$ & $(0.12, 0.08, 0.65)$ & $(2.58, 2.87, 0.73)$ \\
P4 &
$(0.12, 0.01, 2.04)$ & $(0.12, 0.10, 0.33)$ & $(2.32, 2.32, 0.38)$ \\
\bottomrule
\end{tabular}
\end{table}

Three observations follow. 1)~On P2, P3 and P4 the tripled derivative
creates a fast parasitic mode. The raw response does not show it, but
$\Gamma_2$ does: there $N_2 = 2.06$, $7.93$ and $2.04$, while both
lower bands stay within their limits, so the attribution is
unambiguous. The exchange of derivative gain for proportional and
integral gain is still visible in the returned multipliers: on P2,
$F_d = 0.59$ against $F_i = 2.32$ and $F_p = 2.87$.
2)~Every run ends in the limit cycle of Proposition~\ref{prop:conv},
and feasibility is revisited within a few iterations of the cap of
$200$, the last feasible iterates being $192$--$198$.
3)~The rule uses its margin where the plant allows: the integral
multiplier is more than doubled on every plant, the derivative
multiplier is cut on every plant, and no multiplier lies at a box
bound. The returned counts are well inside the limits, which is where
the alternation of Proposition~\ref{prop:conv} leaves the last feasible
iterate, and not a claim of tightness.

The leakage asymmetry also appears in the runs. At iteration $20$ of
the P3 run, the rise has violated all three bands, with
$\mathbf N = (1.62,\, 1.43,\, 4.20)$. A single reduction of band~$0$,
the only band whose fault is certain, returns
$(0.12,\, -0.03,\, 2.72)$. The violation in band~$0$ clears because its
fault was repaired; the one in band~$1$ clears with it, having been
leakage from the same low-frequency mode; the genuine high-band
violation remains and becomes the next $k_{\min}$. A variant that
reduces every violated band at once would have changed three channels
for this single fault, and on this battery it reaches comparable
tunings. We keep single-band action for the attribution it preserves,
not because a penalty was measured.

%=====================================================================
\section{Limitations}
\label{sec:limits}
%=====================================================================

The claim that the oscillation of a band can be cured by reducing the
gain of that band is derived, and not assumed
(Proposition~\ref{prop:mono}, Remark~\ref{rem:scope}). Several limits
remain.

1)~\fix{A plant outside Assumption~\ref{ass:plant} breaks this claim,
but it does so visibly and safely, because the iteration then stops at a
bound. The sign of $N_k$ is not used anywhere; small negative counts, as
for P2 in Table~\ref{tab:battery}, are endpoint artifacts.}

2)~Noise enters the differenced coordinates directly. The guard
excludes settled noise under Remark~\ref{rem:noise}, but it does not
exclude broadband noise inside the transient. Noisy installations need
filtering, which lies outside the scope of this paper.

3)~Each iteration needs one settled window; a disturbance corrupts one
test, not the procedure. 4)~Single-band action is the price of
attribution. A stable start that
violates all bands converges through a sequence of single corrections.
Take P4, started from three, two and four times the initial $K_i$,
$K_p$ and $K_d$ of Table~\ref{tab:battery}. All three bands are
violated at the start, at $\mathbf N = (3.91, 4.02, 4.01)$, and first
feasibility takes $32$ iterations. The \new{graded} backoff makes no
attribution \new{and assumes the top of the range}, and the screen
reserves it for divergence.

5)~Leakage can misdirect a correction when the true source sits just
below its own threshold; the rule then corrects itself within a few
iterations.

6)~A multiplier resting at a bound is an outcome to report rather than
tune through, and the counts tell the cases apart. A multiplier at
$F_{\min}$ while its band still rings blocks the cascade: the violation
of $k_{\min}$ never clears, so no band above it is ever examined, and
for $k_{\min} = 0$ the iteration stops at an infeasible fixed point,
which no run reported here reached.
Proposition~\ref{prop:conv} then gives determinism and boundedness but
not convergence, because its hypothesis excludes this saturation.
Widening the box is the first remedy, although no box can be certified
wide enough in advance. \new{Multipliers at $F_{\max}$ with every count
quiet carry the opposite message: not a box that is too small, but a
plant with no finite stability boundary
(Remark~\ref{rem:ladder}), for which no box is wide enough. P4 from
unit gains rests $K_d$ at the ceiling with every band within limits,
which is this benign case.}

The destination is a fast, well-damped tuning and not an optimum; where
a particular criterion matters, it is a sound starting point for
refinement against it.

%=====================================================================
\section{Conclusion}
%=====================================================================

Three portraits make under-damping visible band by band, a winding
count turns it into a number, a settling-anchored guard makes that
number independent of the observation window, and one triangular rule,
reducing only the band whose fault is certain, turns it into a
decision. The turn indices carry no units and no clock, so fixed
dimensionless constants serve every plant tested. The iteration is
deterministic and bounded without further conditions. A violated count
is itself the evidence that the loop is under-damped, which is the
condition under which reducing that band helps, so the plant class
supplies the monotonicity wherever the rule corrects a gain. Under it
the iteration reaches a bounded neighbourhood of the boundary of the
well-damped set.

\new{\spin} needs no model, and it shows how it reaches its result: an
operator who watches the three portraits sees every decision the
algorithm makes. More generally, a closed-loop record can be read
through invariants of its trajectory instead of through an identified
model, and tuning then becomes the problem of driving three winding
numbers below three fixed fractions of a turn. \rev{The setting here is
the plainest practical one: a fixed step, a noise-free analysis, and
one set of defaults. What the iteration seeks asymptotically, how a
diminishing step would reach it, and how measurement noise should be
handled are left for follow-up work.}

\section*{CRediT authorship contribution statement}
\textbf{Daniel Pachner:} Conceptualization, Methodology, Software,
Writing -- review \& editing. \textbf{Pavel Otta:} Software, Validation,
Visualization, Writing -- original draft. \textbf{Ji\v{r}\'{\i}
Dost\'al:} Investigation, Writing -- review \& editing.
\textbf{Vladim\'{\i}r Havlena:} Supervision, Writing -- review \&
editing.

\section*{Declaration of competing interest}
The authors declare that they have no known competing financial
interests or personal relationships that could have appeared to
influence the work reported in this paper.

\section*{Data availability}
Reference implementation of \spin, with the interactive demonstration
that reproduces results reported here:

\noindent{\raggedright\url{https://robopid-simulator.uceeb.cvut.cz/}\par}

% --- Elsevier requires this section at the end of the manuscript,
% --- immediately before the reference list. Edit the tool name and the
% --- reason so that they describe what was actually done; the
% --- declaration does not cover basic spelling, grammar or reference
% --- checkers, and it does not cover AI tools used for data analysis or
% --- for producing the simulation code.
\section*{Declaration of generative AI and AI-assisted technologies in
the manuscript preparation process}
During the preparation of this work the authors used Anthropic Claude in
order to improve the readability and the language of the manuscript and
to check the consistency of the notation. After using this tool, the
authors reviewed and edited the content as needed and take full
responsibility for the content of the published article.

\bibliographystyle{elsarticle-num}

\end{document}